\documentclass[%
 reprint,
 superscriptaddress,
 amsmath,amssymb,
 aps,
]{revtex4-2}

\usepackage{graphicx}% Include figure files
\usepackage{dcolumn}% Align table columns on decimal point
\usepackage{bm}% bold math
\usepackage{hyperref}% add hypertext capabilities
\usepackage{physics}% bra, ket
\usepackage{xcolor}
\usepackage{mathtools}
\usepackage{amssymb}
\usepackage{color,soul}
\usepackage{comment}

\newcommand{\maria}[1]{{\color{magenta}#1}}
\newcommand{\minor}[1]{{\color{magenta}#1}}

\newcommand{\rvline}{\hspace*{-\arraycolsep}\vline\hspace*{-\arraycolsep}}

\usepackage[makeroom]{cancel}
\usepackage{dsfont}
\usepackage{amsmath,amsfonts}
\newtheorem{theorem}{Theorem}
\newtheorem{lemma}[theorem]{Lemma}

\newtheorem{definition}[theorem]{Definition}

\newtheorem{proposition}[theorem]{Proposition}

\newenvironment{proof}{\medskip\noindent\textbf{Proof.}}{\hfill$\blacksquare$\medskip}
\newenvironment{proof-of}[1]{\medskip\noindent\textbf{Proof of {#1}.}}{\hfill$\blacksquare$\medskip}

\begin{document}

% \preprint{APS/123-QED}

\title
%{ugly title but main idea is: Simulating many-body Hamiltonians in 1D with more local Hamiltonians of the same dimension\\{\color{white}lala} \\another suggestion: Quantum simulation on non-scalable platforms\\
{Simplifying the simulation of local Hamiltonian dynamics}
% Force line breaks with \\
% \thanks{A footnote to the article title}%

\author{Ayaka Usui}
\email{ayaka.usui@uab.cat}
 \affiliation{Departament de F\'{i}sica Qu\`{a}ntica i Astrof\'{i}sica, Institut de Ci\`{e}ncies del Cosmos (ICCUB), Facultat de F\'{i}sica, Universitat de Barcelona, Mart\'{i} i Franqu\'{e}s, 1, E08028 Barcelona, Spain}
 \affiliation{Grup d'\`{O}ptica, Departament de F\'{i}sica, Universitat Aut\`{o}noma de Barcelona, 08193 Bellaterra, Spain}

 \author{Anna Sanpera}
 \email{anna.sanpera@uab.cat}
 \affiliation{F\'{\i}sica Te\`{o}rica: Informaci\'{o} i Fen\`{o}mens Qu\`{a}ntics, %
	Departament de F\'{\i}sica, Universitat Aut\`{o}noma de Barcelona, 08193 Bellaterra, Spain}
\affiliation{ICREA, Pg. Llu\'is Companys 23, 08010 Barcelona, Spain}

 \author{Mar\'ia Garc\'ia D\'iaz}
 \email{maria.garcia.diaz@upm.es}
 \affiliation{Departamento de Matem\'atica Aplicada a la Ingenier\'ia Industrial,
Universidad Polit\'ecnica de Madrid, 28006 Madrid, Spain}

\date{\today}

\begin{abstract}

Local Hamiltonians, $H_k$, describe non-trivial $k$-body interactions in quantum many-body systems. Here, we address the dynamical simulatability of a $k$-local Hamiltonian by a simpler one, $H_{k'}$, with $k'<k$, under the realistic constraint that both Hamiltonians act on the same Hilbert space.
When it comes to exact simulation, we build upon known methods to derive examples of $H_k$ and $H_{k'}$ that simulate the same physics. We also address the most realistic case of approximate simulation. There, we   
%where we prove 
upper-bound the error up to which a Hamiltonian can simulate another one, regardless of their internal structure, and show an example suggesting that the accuracy of a $(k'=2)$-local Hamiltonian to simulate $H_{k}$ with $k>2$ increases with $k$.
% prove, by means of an example, that the accuracy of a $(k'=2)$-local Hamiltonian to simulate $H_{k}$ with $k>2$ \maria{is likely to} increase with $k$.
%\anna{This leads %us to an interesting example showing
% that the accuracy of a $(k'=2)$-local Hamiltonian to simulate $H_{k}$ with $k>2$ \maria{is likely to} increase with $k$.
%worst precision with which a certain 2-local Hamiltonian  simulates a target Hamiltonian increases while the locality of the latter decreases. 
Finally, we propose a method to search for
%, under some restrictions, a $H_{k'}$ 
the $k'$-local Hamiltonian that simulates, with the highest possible precision, the short time dynamics of a given $H_k$ Hamiltonian.
%the short-time regime. }    

\end{abstract}

\maketitle

\section{Introduction}\label{sec:intro}

A quantum simulator (QS) aims at replicating the  physics of another quantum system, the properties of which are extremely difficult to obtain~\cite{Feynman1982}.
%classically. 
Currently, most of the existing QS are analog, i.e.,  naturally realise the physics of the target system and mimic its properties. Moreover, they are highly specialized 
experimental platforms acting as  single-purpose simulators, and
%their devices are highly specialized for single purpose. Moreover,}
% Currently, most of the existing (analog) QS, are highly specialized 
%experimental platforms acting as  
% single-purpose simulators, and 
all of them, independently of the physical architecture on which they are implemented, face some problems regarding applicability, scalability, complexity, state preparation, control, and measurement (see e.g. \cite{altman2021quantum} and references therein). 
It is desirable and expected that in the forthcoming years, some analog QSs will evolve into highly flexible devices, %become more flexible 
and also digital QSs, i.e., circuit-based quantum computers, will be further developed.
%handle quantum circuits and apply quantum gates for programming a target Hamiltonian or time evolution. 
Both classes of QS will be capable of simulating different quantum systems, preparing specific quantum states on demand, or analyzing their dynamics. % This type of QSs include}
% some QSs will evolve into highly flexible programmable devices capable of simulating different quantum systems, preparing specific quantum states on demand, or analyzing their dynamics. 
This also includes 
the current NISQ (Noise Intermediate Scalable Quantum) platforms such as superconducting qubits, cold atoms and ions, Rydberg atoms, etc.
Such versatility demands the development of theoretical frameworks capable of assessing the possibilities and applications of QS in areas as diverse as condensed-matter, high-energy physics, quantum chemistry, quantum gravity, out-of-equilibrium quantum physics or others~\cite{lewenstein2012, georgescu2014,altman2021quantum,Cirac2012,Johnson2014}.

%While fully fault-tolerant operations cannot be implemented, NISQ (noisy, intermediate scale, quantum) devices demand the development of theoretical frameworks to assess the possibilities offered  by (analog) QS.   

%arises from its broadness  range of potential applications it offers in areas such as condensed-matter physics, high-energy physics or quantum chemistry \cite{georgescu2014}, in an era where fully fault-tolerant operations have not yet taken over NISQ (noisy, intermediate scale, quantum) devices. 

As it 
%can be noticed from 
is reflected in the literature~\cite{Lloyd1996,georgescu2014,suba2016,Cubitt2018,Keever2023}, there are different approaches to define %ways of defining 
what are the general conditions that a QS should fulfill.
%frameworks should be. 
Those are obviously linked to the particular requirements needed to achieve a given goal, which in turn can also depend on the specific platform on which the QS will be implemented.
%definition of what it means for a quantum %system to simulate another one---which %usually 
%which also depends on the specific limitations imposed by the accessible experiments. 

On the one hand, 
%variouos works
%quite remarkably, 
there is the strongest notion of a ``universal'' QS: 
a quantum many-body system simulates another one if it is able to reproduce its entire physics, i.e., it replicates its eigenstates, full energy spectrum, time evolution, any local noise process, correlation functions, observables, and thermal properties. This is the path developed in a series of very remarkable works ~\cite{Cubitt2018,Kohler2020,kohler2022,zhou2021strongly} based on a perturbative approach (perturbative gadgets) proposed initially in~\cite{kempe2005}. 
%are based on a quite demanding notion of simulation: a quantum system simulates another quantum many-body system if it is able to reproduce its entire physics, i.e., its eigenstates, full energy spectrum, time evolution, any local noise process, correlation functions, observables, and thermal properties. In \cite{Cubitt2018}, for instance, they built upon the perturbative gadgets proposed by 
Within this approach the authors  proved that the physics of any quantum  many-body system can be replicated by certain 2-local spin-lattice models. The latter should be embedded in a Hilbert space significantly larger than that of the system to be replicated.  Even a simple translationally  invariant spin chain in 1D  turns out to be a \textit{universal simulator} in the above sense~\cite{Kohler2020}.
Nevertheless, such an approach is not constructive and comes with the additional challenge of requiring a scale-up of the simulation platforms, either by increasing the local dimension of the many-body system acting as a simulator and/or demanding an exquisite control over the interactions, which cannot be implemented in current and/or near future experiments.
%perturbative gadgets require the addition of $k$ ancillary qubits for each $k$-body interaction to be reproduced \cite{suba2016}.   
On the other hand, there are less demanding definitions of quantum simulation where the spell is put on simulating either the ground state or the low energy physics of a complex quantum many-body system or some restricted dynamics. Simulating the dynamics of many-body quantum systems is, in general, a very difficult task for a classical or even for a NISQ device. Some advances have been recently proposed in this direction. For example, the purpose of \cite{suba2016,gietka2021} is to find a QS, whose time evolution in a restricted subspace is identical to that of the target system. In \cite{suba2016}, the goal is to increase the locality of the QS, %(small $k'$) (that is, to reduce $k'$).
and this is achieved by requiring an additional qubit on each $k$-body interaction. On the contrary,  in \cite{gietka2021}, the dimension of the QS does not increase, but the locality of both Hamiltonians remains the same. 
%is the same as the one of the target system but the locality of the  QS is the same as the one of the target quantum  many-body systems.
%,  \cite{suba2016} aims for a simulator Hamiltonian that is more local than the target one, a requirement that is not pursued by \cite{gietka2021}.  

%\anna{ Despite a vast majority of paradigmatic models are 2-local, this is not always the case. For example lattice gauge theory (LGT) models involve $k\ge 3$ interactions, making  their implementation in present QS extremely hard or impossible. The question we address here is: when and how a complex $k$-local Hamiltonian can be simulated with a $k'$-local Hamiltonian with $k'<k$? We will  follow an ''approximation Hamiltonian'' method, where we aim at providing protocols and algorithms able to accurately reproduce the ground state or the low energy spectrum of a complex Hamiltonian by a simpler ''experimentally'' feasible  ($k=2$) one, releasing ourselves from the demand of accurately mimic the \emph{whole spectrum} of the complex $k$-local Hamiltonian. 

Our work is aligned with the latter ones in that our interest is placed on simulating dynamics of a target $k$-local Hamiltonian using a $k'$-local Hamiltonian for $k'<k$, i.e., the simulator is more local than the target. 
%simulator Hamiltonians that are more local than the target ones, in light of the  pressing difficulty of experimentally realizing many-body interactions. 
We do not attempt to embed the QS in a Hilbert space larger than the one associated to the target Hamiltonian, given that adding qudits to an experimental platform remains a challenging task. Also, for simplicity, we restrict our Hamiltonians to be 1D. 
% we release the constraint of embedding the QS in a Hilbert space larger than the one associated to $H_T$ (the target Hamiltonian), 
Our purpose is to investigate the simulation of quantum many-body dynamics with the following two aims: simplifying complexity while maintaining the local dimension of the parties and the total dimension of the Hamiltonian.

Our work is organized as follows. After presenting some 
%required 
preliminary concepts, in Section \ref{exact} we focus on exact dynamics and show, by means of an example, 
%that it is sometimes possible to find two  Hamiltonians,the 
how some dynamics of a target Hamiltonian, $H_T$, 
%and the simulator one 
can be replicated exactly by a simulator, $H_{QS}$, with simplified interactions %different localities living 
in the same Hilbert space.
%which %simulate
%\anna{displaying the same dynamical behaviour in a given  range of parameters}. 
In this example both Hamiltonians commute, which is a quite exceptional condition. Nevertheless, when considering non-commuting Hamiltonians, i.e., $[H_T,H_{QS}]\neq 0$, 
it is still possible to reproduce some dynamical behaviour 
%can always  through another example that the previous kind of  simulation can be performed 
on a certain energy subspace, providing that both Hamiltonians have some eigenstates in common. It is hard to search for such shared subspaces by diagonalizing large Hamiltonians, and thus we present a generic upper bound to the dimension of such subspaces. In Section \ref{approx}, we consider the most realistic case of non-exact simulation. There, we provide an upper bound to the error up to which any Hamiltonian is able to simulate another one, irrespectively of their structure. In an attempt to understand how such upper-bound depends on locality, we present an interesting case suggesting
% showing 
that the precision with which a certain 2-local QS simulates a broad family of target Hamiltonians
% \maria{is more likely to} 
increase while the locality of the target Hamiltonian decreases. This suggests that, at least in certain cases
and in reduced dimensionalities, a $k$-local $H_T$ 
%less local Hamiltonians 
could be simulated better by a 2-local $H_{QS}$ as $k$ increases,
%the larger the $k$ is, 
or in other words, the less local $H_T$ is.  
%more local ones.  
Finally, after showing that the worst-case error of simulation decreases when considering short times as expected, 
% after showing that the worst-case error of simulation \maria{is more likely to} decrease when considering short times, 
we 
%put forward a program 
develop a simple algorithm to find the $k'$-local Hamiltonian that best simulates a given \mbox{$k$-local} Hamiltonian in the short-time regime. 
%\anna{I am not sure about that, it weakness the paper for me.}
%For ease of exposition, we defer all the proofs to the Appendix.

\begin{comment}
intro...

It has been reported that two distinct Hamiltonians could produce the same dynamics and shown that as an example one-axis twisting, which typical dynamics appear from all-with-all interactions, is generated by using nearest neighbouring interactions along with a staggered field~\cite{}. Based on the work, we dig into the question of how simulation succeeds with a different Hamiltonian from the target one. Specifically, here we ask whether simulation is possible with fewer local Hamiltonians. Also, we restric

Recently, it has been proved that Heisenberg model can be universal and simulate the same physics for any target Hamiltonian, but one has to allow the system size of the model grow~\cite{}. We focus on simulation in the same dimension and restrict the system size of simulators to that of the target. 

In this work, we aim at simulating the dynamics of many-body Hamiltonians with fewer local Hamiltonians.  We employ the idea of the connector introduced in Ref.~\cite{}, where 

...

review of simulation
- feynmann
- refs in cubiit
- commuting

notation
- dynamics
- time indepednet
- HQS, HT
- HQS less(?) locality than HT
\end{comment}

\section{Preliminary concepts}
%\anna{ you mean something different I believe.. Physically significant Hamiltonians tend to be local.}
Most of quantum many-body systems are described by local Hamiltonians, and a vast majority of them only involve 2-body interactions. However, there are well known examples of relevant many-body Hamiltonians  (lattice gauge theories, Kitaev topological models, etc.) where interactions are not 2-body and whose properties are very difficult to analyze. To fix concepts, we denote by 
%Let 
${\cal{H}}_N$  an $N$-dimensional Hilbert space and by $\mathcal{B}({\cal{H}}_N)$  the set of  its bounded operators. An $n$-qudit \emph{$k$-local} Hamiltonian $H\in \mathcal{B}({\cal{H}}_{d^n})$ is of the form $H=\sum_{i=1}^{M}h_i$, where $h_i$ is a Hamiltonian acting non-trivially on at most  $k$ parties, and $M$ is some positive integer. An $n$-qudit $k$-local Hamiltonian with $k=n$ is dubbed \emph{nonlocal}. Furthermore, a $k$-local Hamiltonian is said to be \emph{more local} than a $k'$-local Hamiltonian if $k<k'$. 
Note that  $k$-locality does not
%$k$-local interactions  do not 
necessarily imply
%mean 
interaction among $k$ next neighbouring objects, and that there is no restriction on spatial locality.

We denote by $H_T\in \mathcal{B}({\cal{H}}_{N_{T}})$ the Hamiltonian describing the target quantum system, the time evolution of which is to be replicated. The system achieving such quantum simulation will, in turn, be described by a Hamiltonian $H_{QS}\in \mathcal{B}({\cal{H}}_{N_{QS}})$. As discussed previously, here we impose that $N:=N_T=N_{QS}$. %and that $H_{QS}$ is more local than $H_T$.  
We hereby introduce our precise definition of simulation ($\hbar=1$ is set throughout the whole manuscript): 

\begin{definition}
A Hamiltonian $H_{QS}$ $\epsilon$-simulates %dynamically 
a target Hamiltonian $H_T$  at state $\ket{\psi}$ and time $t$ if
\begin{equation} \label{def_simul}
  |\bra{\psi}e^{itH_{QS}}e^{-itH_{T}}\ket{\psi}|\geq 1-\epsilon,  
\end{equation} 
with $\epsilon\in(0,1]$.
\end{definition}

Notice that $\epsilon=0$ implies that the fidelity is 1, meaning that
%Notice that $\epsilon=0$, \ayaka{that is the fidelity is 1,} means that 
$H_{QS}$ simulates $H_T$ exactly in a subspace $\mathcal S\subset \mathcal H_N$, where $|\psi\rangle$ belongs. 
%\ayaka{Minimizing $\epsilon$} 
Maximizing $\epsilon$ over initial states and times informs us about the minimal performance of a given $H_{QS}$ at simulating a certain $H_T$. The following definition can  be put forward:

\begin{definition}
    The worst-case fidelity of simulation for given Hamiltonians $H_{QS}$ and $H_T$ is defined as
    \begin{equation}
        \min_{\psi,t} |\bra{\psi}e^{itH_{QS}}e^{-itH_{T}}\ket{\psi}|.
    \end{equation}
\end{definition}

%Also, if $\mathcal{S}$ is the trivial subspace $\mathcal{S}=\mathcal{H}_N$ then condition (\ref{def_simul}) implies that the simulation is feasible in the whole space.

% (Fidelity is usually the absolute square, and so this condition maybe should be the absolute square but the absolute value may be easier to see... we can change anytime)

%\begin{definition}
%A Hamiltonian $H_{QS}$ $\epsilon$-simulates a target Hamiltonian $H_T$ in subspace $\mathcal{S}\subset  \mathcal{H}_N$ at state $\ket{\psi}$ and time $t$ if
%\begin{equation} \label{def_simul} |\bra{\psi}e^{itH_{QS}^{(\mathcal{S})}}e^{-itH_{T}^{(\mathcal{S})}}\ket{\psi}|\geq 1-\epsilon,  
%\end{equation} with $\epsilon\in[0,1]$, $\ket{\psi}\in \spn(\mathcal{S})$ and $H_X^{(\mathcal{S})}$ the projection of $H_X$ onto $\mathcal{S}$.
%\end{definition}

\section{Exact simulation}\label{exact}

In this section we address the scenario of error-free simulation, i.e., when $\epsilon=0$. We illustrate this scenario by showing examples of both commuting and non-commuting Hamiltonians which generate the same dynamics despite their different locality. 

%\subsection{In the whole space $\mathcal{H}_N$}\label{exact_wholespace}

\subsection{Commuting Hamiltonians}
%Let us start by considering the case when $\mathcal{S}=\mathcal{H}_N$. 
%Let us start by considering that $H_T$ and $H_{QS}$ share the full space and $\ket{\psi}\in \spn(\mathcal{H}_N)$. 
%For the sake of brevity, we will simply denote $H_X:=H_X^{(\mathcal{H}_N)}$. 

Using the Baker-Campbell-Hausdorff (BCH) formula, one can write condition (\ref{def_simul}) as
\begin{equation}\label{def_exactsimul}
    |\bra{\psi}e^{i\tilde{H}(t)}\ket{\psi}|= 1,
\end{equation}
with $\tilde{H}(t):=t(H_{QS}-H_T)+\frac{it^2}{2}[H_{QS},-H_T]+...$, where the remaining terms involve higher order commutators of $H_{QS}$ and $H_T$. In order for the BCH series to converge absolutely in this case, note that $H_{QS}$ and $H_T$ must fulfill $\norm{H_{QS}}+\norm{H_T}<\pi/t$, for a given submultiplicative norm $\norm{\cdot}$ \cite{Mityagin1990,Blanes2004}. Given a target Hamiltonian $H_T$, one can now ask for which %quantum simulator Hamiltonians 
$H_{QS}$, initial states $\ket{\psi}$ and times $t$, condition (\ref{def_exactsimul}) is fulfilled. An answer was provided in \cite{gietka2021}, which we briefly sketch in what follows: consider $H_{QS}$ such that $[H_{QS},H_{T}]=0$, and define the connector $h:=H_{QS}-H_{T}$. Now write the initial state in the basis $\{\ket{\phi_j}\}_j^N$ which diagonalizes $h$: \mbox{$\ket{\psi}=\sum_{j=1}^N a_j \ket{\phi_j}$}, where $a_j=\braket{\phi_{j}}{\psi}$. Then condition (\ref{def_exactsimul}) becomes
$\left|\sum_{j=1}^{N} |a_j|^2 e^{i\lambda_j t} \right|=1$, with $\lambda_j$ the $j$-th eigenvalue of $h$. If $a_j\neq 0$ for each $j\in[1,N]$, the previous condition holds if $h$ is fully degenerate, implying that $H_{QS}$ is simply $H_T$ plus some multiple of the identity. In order to avoid this trivial scenario, we naturally require that $\ket{\psi}$ is only spanned by the degenerate eigenvectors of $h$, i.e., $a_j\neq 0$ for all $j$ such that $\ket{\phi_j}$ is degenerate  and $a_j=0$ otherwise. All in all, any Hamiltonian $H_{QS}$ that commutes with a target Hamiltonian $H_T$ is able to simulate $H_T$ exactly at any time $t$  and at any initial state spanned by the degenerate eigenvectors of the corresponding connector. Note that $e^{ith}$ converges absolutely, as does the exponential of any square matrix.

%Since $H_{QS}$ and $H_T$ share all their eigenvectors, we have $[H_{QS},H_{T}]=0$. This leads to $\tilde{H}(t)=t(H_{QS}-H_{T})$, and let us define $h\equiv H_{QS}-H_{T}$ and call it connector~\cite{gietka2021}. Also, any initial state $\ket{\psi}$ is written as 
%\begin{align}
%    \ket{\psi}
%    &=
%    \sum_{j=1}^N a_j \ket{\phi_j}
%\end{align}
%with $a_j=\braket{\phi_{j}}{\psi}$. Some of the coefficients $a_j$ may be zero, i.e. the initial state may be constructed by only part of the eigenvectors $\ket{\phi_j}$. Therefore, modify the index $j$ such that $a_j\neq0$ for $j=1,2,\ldots,N_0$ and $a_j=0$ otherwise. 
% and therefore, $\braket{\psi_0}{\psi_0}=\sum_{j=1}^{N_0} |a_j|^2=1$. 
%The condition~\eqref{def_exactsimul} can be rewritten as
%\begin{align}
%    \left| \sum_{j=1}^{N_0} |a_j|^2
%    \mathrm{e}^{ih_j t} \right|
%    =
%    1
%\end{align}
%with $h_j$ the $j$-th eigenvalue of $h$. The above condition is satisfied always if the eigenvalue $h_j$ is degenerate for $j=1,2,\ldots,N_0$, i.e. $h_j=h$.

%From the above, the simulator Hamiltonian $H_{QS}$ that commutes the target Hamiltonian $H_T$ succeeds in the simulation if the initial state $\ket{\psi}$ is spanned where the connector $h$ is degenerate. This is the condition the previous work~\cite{gietka2021} has proposed. 

As an example of this approach, in \cite{gietka2021} it was shown that an infinite-range-interaction (2-body all with all interactions) Hamiltonian can be simulated using a nearest-neighbor-interaction model with a staggered field. Notice that, in this case, both $H_{T}$ and $H_{QS}$ are 2-local Hamiltonians. Here, we demonstrate, by means of a toy model, that the same approach can be employed to find instances of Hamiltonians \emph{with different localities} that also lead to the same dynamics.

%As an example, it has been showed that a neighbouring 2-body interacting spin system can produce one-axis dynamics, which are well studied dynamics that the infinite range 2-body interactions generate. 

%Here, we show another example of two commuting Hamiltonians showing the same dynamics, but note that in this case 2-local interactions and 3-local interactions produce the same dynamics. 

As an ideal terrain to understand such dynamics, let us consider a simple case in which only 4-qubits are involved. Let the 3-local Hamiltonian
\begin{equation} \label{eq:H3local}
 H_T=\sum_{j=1}^4(J_3\sigma_z^{j}\sigma_z^{j+1}\sigma_z^{j+2}+h_x\sigma_x^{j}),   
\end{equation}  
%\begin{equation}
% H_T=\sum_{j=1}^4(\sigma_z^{j}\sigma_z^{j+1}\sigma_z^{j+2}+\sigma_x^{j}),   
%\end{equation} 
describing a one-dimensional system with periodic boundary condition, be our target Hamiltonian. By engineering  the (2-local) Heisenberg XYZ model 
\begin{equation} 
    H_{QS}=\sum_{j=1}^4(J_x\sigma_x^j\sigma_x^{j+1}+J_y\sigma_y^j\sigma_y^{j+1}+J_z\sigma_z^j\sigma_z^{j+1})
\end{equation}
we are able to  obtain a QS Hamiltonian which commutes with $H_T$ and such that the corresponding connector has some degeneracy. Notice that \mbox{$\norm{[H_T,H_{QS}]}^2_{\text{HS}}=|J_3(J_x-J_y)|^2/2+8|h_x(J_y-J_z)|^2$}, where $\norm{X}_{\text{HS}}=\left(\sum_{i,j=1}^m |X_{ij}|^2\right) ^{1/2}$ is the  Hilbert-Schmidt (HS) norm of an $m$-dimensional matrix $X$. Therefore, choosing $J:=J_x=J_y=J_z$ (i.e., the Heisenberg XXX model), commutativity between $H_T$ and $H_{QS}$ is granted. By tuning $J$ (see that $J_3$ and $h_x$ are fixed) and comparing the eigenvalues of the resulting $H_{QS}$ with those of $H_T$, it is possible to create degeneracy in the connector. The evolution under $H_T$ of any initial state belonging to the subspace where the connector is degenerate can thus be exactly reproduced by the engineered $H_{QS}$. A particular example is displayed in Appendix \ref{app:example1}.
Clearly, the same approach can also be taken for higher spatial dimensional systems, where spins have a larger number of next neighbours and thus exhibit a greater richness of interactions. 
%the interactions are more rich. 
However, since this approach requires diagonalization, implementing it in large systems is not viable. % we display such examples.
% example where $J_3=h_x=1$ and $J\approx -0.99$. 
% The chosen initial state is spanned by the first two degenerate eigenstates of the connector.  
%\ayaka{
%However, the approach used there is possible to do only for small systems, and it is infeasible to adjust it to large systems. 
% \maria{Nevertheless, in Proposition \ref{systematic}} 
As a more systematic way, we introduce a method to search for the best simulator Hamiltonian in any dimension and with fixed locality in Proposition \ref{systematic}
%fixed in Sec.~\ref{approx} 
by allowing for an error $\epsilon$ in the fidelity of simulation, and by restricting to the short-time regime.

\subsection{Non-commuting Hamiltonians}

Even if 
$[H_{QS},H_T]\neq 0$, 
%does not commute with the target Hamiltonian $H_{T}$, 
simulation is feasible in some cases. If both Hamiltonians share part of their eigenstates, it naturally follows from the previous approach  that simulation can still be performed on the energy subspace spanned by those eigenstates. 
%as both Hamiltonians commute in that subspace. 
Let the set of shared eigenstates be $\Theta=\{\ket{\varphi_j}\}_{j=1}^{N_\Theta}$ (note that $N_\Theta$ depends on the basis on which each Hamiltonian is expressed) and let  the projected connector be \mbox{$h^{(\Theta)}:=\sum_{i,j=1}^{N_\Theta} \ketbra{\varphi_i}{\varphi_i} h \ketbra{\varphi_j}{\varphi_j}$}. 
% \mbox{$h^{(\Theta)}:=\sum_{j=1}^{N_\Theta} \ketbra{\varphi_j}{\varphi_j} h \ketbra{\varphi_j}{\varphi_j}$}. 
Then the simulatable states $\ket{\psi}$ will be those living in the subspace where $h^{(\Theta)}$ is degenerate.  %\maria{Note that, again, $e^{ith^{(\Theta)}}$ converges absolutely.}

We illustrate with an example how two non-commuting Hamiltonians with different localities may generate the same dynamics in Appendix~\ref{app:example2}. There again, we consider the $H_T$ of the previous section as our 3-local target Hamiltonian. Our QS is now a (2-local) Heisenberg XXX model with 1-body terms,
\begin{align} \label{eq:tildeHQS}
     \tilde{H}_{QS}
     &=
     J \sum_{j=1}^4
     \left(
     \sigma_x^j\sigma_x^{j+1}
     +
     \sigma_y^j\sigma_y^{j+1}
     +
     \sigma_z^j\sigma_z^{j+1}
     \right)
     \nonumber\\
     &\quad
     +b_x\sum_{j=1}^4\sigma_x^j
     +b_y\sum_{j=1}^4\sigma_y^j
     +b_z\sum_{j=1}^4\sigma_z^j,  
\end{align}
where these 1-body terms prevent $\tilde{H}_{QS}$ from commuting
with $H_T$.

% which does not commute with $H_T$ due to the presence of the 1-body interactions. By properly adjusting the free parameters of the model, $J$ and $\vec{b}=(b_x,b_y,b_z)$, it is possible to exactly replicate the dynamics of $H_T$, as displayed in   Fig.~\ref{fig:spinob_partiallyshared}, where the nontrivial independent spin component dynamics is also plotted.
%for the corresponding plots.  %By changing $J$, we make the connector $h$ degenerate. Figure~\ref{fig:spinob_partiallyshared} displays a successful case and shows that the spin components at each spin in the dynamics. 

% \begin{figure}
% \centering
% \includegraphics[width=0.95\linewidth]{spinob_partiallyshared.png}
% \caption{Evolution of the expected value of the spin components at each spin under $H_T$ and the fidelity $|\bra{\psi}e^{it\tilde{H}_{QS}}e^{-itH_{T}}\ket{\psi}|$, where $J=\sqrt{17}/4$, $b_x=-4$, $b_y=0$, and $b_z=1$ are set. }
% %See Appendix \ref{H3matrix} for the chosen initial state $\ket{\psi}$.}
% \label{fig:spinob_partiallyshared}
% \end{figure}

 Yet a far more complex question is whether two given non-commuting high-dimensional Hamiltonians give rise to the same dynamics. The answer is generally negative, since the subset of Hamiltonians with only one common eigenstate is already of measure zero \cite{garciadiaz2022}.
 % \footnote{On a similar note, in Appendix~\ref{necessary_eig} we derive a necessary condition for a state to be an eigenstate of two Hamiltonians with different localities.}. 
 Moreover, computing the shared eigenstates entails diagonalizing such large Hamiltonians, which costs computational resources. Below, we derive an upper bound to the maximum number of such shared eigenstates, the calculation of which does not require diagonalization or computationally demanding techniques: 
 % Moreover, computing the shared eigenstates for each particular case is computationally unfeasible. 
 % Nevertheless, an upper bound to the maximum number of such common eigenstates can be derived: 
\begin{lemma}\label{number_shared_ev}
The maximum number $r$ of shared eigenstates of two $N$-dimensional non-commuting Hamiltonians $H_T$ and $H_{QS}$ is  bounded as
\begin{equation}
     r \leq N-\left(\dfrac{\norm{[H_T,H_{QS}]}_{\text{HS}}}{\norm{[H_T,H_{QS}]}_2}\right)^{2},
\end{equation}
where % $\norm{X}_{\text{HS}}=\left(\sum_{i,j=1}^n |X_{ij}|^2\right) ^{1/2}$ is the  Hilbert-Schmidt (HS) norm of an $n$-dimensional matrix $X$ and 
$\norm{X}_2=\max_{\norm{\bf{v}}=1} (\norm{X\bf{v}})$ is the spectral norm of an $n$-dimensional matrix $X$, with $\bf{v}$ an $n$-dimensional vector and $\norm{\bf{v}}=\left( \sum_{i=1}^n|{\bf{v}}_i|^2 \right)^{1/2}$.
\end{lemma}

\begin{proof}
    Consider two arbitrary $N$-dimensional square matrices $A_{\cal{X}}$ and $B_{\cal{X}}$, where $\cal{X}$ is the basis in which they are expressed. Their commutator can be written as $[A_{\cal{X}},B_{\cal{X}}]=0^{\oplus r_{\cal{X}}} \oplus C_{\cal{X}}$, with $r_{\cal{X}}\in \mathbb{N}$, \mbox{$0^{\oplus p}:= \underbrace{0\oplus ... \oplus 0}_{\substack{p}}$}, and $C_{\cal{X}}$ some traceless \mbox{$(N-r_{\cal{X}})$-dimensional} matrix (without loss of generality, we have assumed that $A_{\cal{X}}$ and $B_{\cal{X}}$  share the first $r_{\cal{X}}$ eigenvectors). It then holds that \mbox{$\norm{[A_{\cal{X}},B_{\cal{X}}]}_*=\norm{C_{\cal{X}}}_*$}, for any chosen matrix norm $*$. Let us consider the Hilbert-Schmidt norm, $\norm{\cdot}_{\text{HS}}$, and the spectral norm, $\norm{\cdot}_2$. Now, recall that \mbox{$\norm{X}_{\text{HS}}\leq \sqrt{m} \norm{X}_2$}, with \mbox{$m=\dim X$} \cite{horn2012matrix}. If matrix $X$ is a finite dimensional commutator (thus traceless), then the inequality can get saturated for even $m$ if $X$ is proportional to a diagonal matrix with half of their entries being $1$, and half of them being $-1$.
    For $n$ qubits (thus $m=2^n$), the inequality is saturated when choosing $A_{\cal{X}}=...\mathds{1}\otimes\sigma_x\otimes \mathds{1}...$ and $B_{\cal{X}}=...\mathds{1}\otimes\sigma_y\otimes \mathds{1}...$.
    It then follows that \mbox{$0\leq \norm{[A_{\cal{X}},B_{\cal{X}}]}_{\text{HS}}\leq   \sqrt{N-r_{\cal{X}}} \norm{[A_{\cal{X}},B_{\cal{X}}]}_2$}. Now, the maximum number of shared eigenvectors between $A_{\cal{X}}$ and $B_{\cal{X}}$ is given by 
    \begin{equation}
        r:=\max_{\cal{X}} r_{\cal{X}}\leq  N-\left(\dfrac{\norm{[A_{\cal{X}},B_{\cal{X}}]}_{\text{HS}}}{\norm{[A_{\cal{X}},B_{\cal{X}}]}_2}\right)^{2},
    \end{equation}
 since the considered norms are unitarily invariant. This completes the proof. 
\end{proof}

Note that the calculation of the spectral norm is efficient as it is expressible by a semidefinite program \cite{watrous2009}. %Although this upper bound is not tight, it
This upper bound, which is tight for  even dimension $N$,
 informs about the maximal size that the shared subspace could have, helping to decide whether it is still worth trying to search for the shared eigenstates even by using diagonalization. Moreover, computing the upper bound to $r$ can aid in the processes of figuring out the parameters of the corresponding $H_{QS}$. Suppose one can only prepare initial states \mbox{$\ket{\psi}\in \mathcal{S}\subseteq \mathcal{H}_N$}, where $\dim \mathcal{S}\leq N$. Recall that exact simulation is granted in a subspace of at least dimension $\dim \mathcal{S}^{(\Theta)}$, where $\mathcal{S}^{(\Theta)}$ is the subspace in which $h^{(\Theta)}$ is degenerate. Hence, \mbox{$\dim \mathcal{S}^{(\Theta)}\leq r \leq N-\left(\dfrac{\norm{[H_T,H_{QS}]}_{\text{HS}}}{\norm{[H_T,H_{QS}]}_2}\right)^{2}$}. In order for the dynamics to be simulatable in this case, it is required that $\dim \mathcal{S}\leq \dim \mathcal{S}^{(\Theta)}$. If the current choice of the parameters of $H_{QS}$ leads to \mbox{$N-\left(\dfrac{\norm{[H_T,H_{QS}]}_{\text{HS}}}{\norm{[H_T,H_{QS}]}_2}\right)^{2}< \dim \mathcal{S}$}, it holds that \mbox{$\dim \mathcal{S}>\dim \mathcal{S}^{(\Theta)}$} and therefore  one must search for a different set of parameters.

%one how the size of the shared subspace could be at most and helps one think whether it is worth trying to search for the shared eigenstates even by using diagonalization. 
% Although it does not answer the question, this lemma reveals that  the dimension of the simulatable states, in case they exist, cannot exceed the upper bound. 

In addition, we prove a necessary condition for a state to be an eigenstate of two Hamiltonians with different localities and leave its derivation in Appendix~\ref{necessary_eig}.

\section{Approximate simulation}\label{approx}
We now focus on the realistic scenario of approximate quantum simulation, i.e., the case when $H_{QS}$ simulates $H_{T}$ up to some tolerable error $\epsilon\neq 0$. Contrary to the setting with $\epsilon=0$, looking at condition (\ref{def_simul}) does not  reveal in a straightforward way how both Hamiltonians should be related in order for one to $\epsilon$-simulate the other. Nonetheless, 
%based on a beautiful result derived in \cite{lidar2008}, 
we are able to provide an upper bound to the error up to which such a simulation can be performed:

\begin{theorem}\label{thm_nonexact}
\cite{referee} Every Hamiltonian $H_{QS}$ \mbox{$\epsilon^*$-simulates} any target Hamiltonian $H_T$ at any state $\ket{\psi}$ and time $t$, with 
\begin{equation} \label{eq:worsterror}
    \epsilon^*=\min\left[ 1,\dfrac{t\Delta_h}{2}\right].
\end{equation}
Here, $\Delta_{h}$ is the spectral diameter of the connector, \mbox{$h:=H_{QS}-H_{T}$}, i.e., \mbox{$\Delta_{h}=\max_{ij}|\lambda_i-\lambda_j|$}, where $\{\lambda_i\}_{i=1}^N$ are the eigenvalues of $h$. 
\end{theorem}
\begin{proof}
  First note that $|\bra{\psi}e^{itH_{QS}}e^{-itH_T}\ket{\psi}|=|\bra{\psi}e^{it\tilde{H}_{QS}}e^{-itH_T}\ket{\psi}|$, with $\tilde{H}_{QS}=H_{QS}+c\mathds{1}$, $c\in \mathbb{R}$. Now, it is well established that
  \begin{eqnarray}
      1-|\bra{\psi}e^{it\tilde{H}_{QS}}&e^{-itH_T}&\ket{\psi}|
      \\
      &\leq&|1-\bra{\psi}e^{it\tilde{H}_{QS}}e^{-itH_T}\ket{\psi}|\nonumber \\
      &=&\left|\bra{\psi}\int_t^0 ds\left( \frac{d}{ds}e^{is\tilde{H}_{QS}}e^{-isH_T}\right)   \ket{\psi}\right| \nonumber\\
      &\leq& \left| \int_t^0 ds \norm{\tilde{H}_{QS}-H_T}_2  \right| \nonumber \\
      &=&t \norm{\tilde{H}_{QS}-H_T}_2 \nonumber\\
       &=&t \norm{h+c\mathds{1}}_2.
  \end{eqnarray}

  Therefore, 
  \begin{eqnarray}
      |\bra{\psi}e^{itH_{QS}}e^{-itH_T}\ket{\psi}|&\geq& 1-t\min_c \norm{h+c\mathds{1}}_2 \nonumber \\
      &=& 1-t\frac{\Delta_h}{2},
  \end{eqnarray}
 and the proof is finished.
\end{proof}

From here we see that the worst-case fidelity of simulation is never smaller than $1-\epsilon^*$. We emphasize that the calculation of the error $\epsilon^*$ does not require diagonalization. The spectral diameter $\Delta_h$ is the difference between the largest and the smallest eigenvalue of $h$ and can be obtained by means of a semidefinite program, as discussed later. 
As shown in Appendix \ref{comparison_bounds}, other valid bounds can be derived for the worst-case fidelity of simulation, but they are smaller than the proved $1-\epsilon^*$.
% Neither of the bounds, $b_1(h,t):=\frac{1}{2}\left(e^{t\Delta_{h}}-1\right)$ nor $b_2(h,t):=t\norm{h}_\text{HS}$, can universally upper-bound the other. As shown in Appendix \ref{comparison_bounds}, $b_1(h,t)$ can be smaller than $b_2(h,t)$ at certain instances of $h$ and $t$, while being greater than it at other instances. 
%For instance, in the regime $t\Delta_h\leq 1$, it holds that $0\leq b_1(h,t)\leq t\Delta_h\leq 2t\norm{h}_\infty\leq 2t\norm{h}_\text{HS}=2b_2(h,t)$, where $\norm{\cdot}_\infty$ is the operator norm \cite{lidar2008}. This does not preclude $b_1(h,t)$ from being smaller than $b_2(h,t)$ at certain instances of $h$ and $t$, while being greater than it at other instances.

% Let us now focus on scenarios where $b_1(h,t)<b_2(h,t)$. Here 
Notice that $1-\epsilon^*$ is greater than zero if and only if \mbox{$t\Delta_{h}<2$}. 
% \mbox{$t\Delta_{h}<\log 3$}. 
In this regime, for fixed time $t$, we see that the smaller $\Delta_{h}$, the higher $1-\epsilon^*$. %Equivalently, the smaller $\Delta_{h}$, the larger the simulation time for a fixed accuracy in the fidelity.
Now note that since $H_{QS}$ is a $k'$-local Hamiltonian and $H_T$ is a $k$-local Hamiltonian (with $k'<k$), $h$ is a proper $k$-local Hamiltonian. One could ask at this point how the spectral diameter of a $k$-local Hamiltonian depends on $k$. %, which would suggest that $k$-local Hamiltonians could, in general, be simulated to a better precision. 
To answer this open question, research in the direction of  \cite{Harrow2017} would be required. They study the extremal eigenvalues of $k$-local Hamiltonians (with $k={\mathcal{O}}(1)$) acting on $n$ qubits, such that each qubit participates in at most $l$ terms. Instead of showing, as they do, how such extremal eigenvalues change with the interaction $l$, one would need to examine how these vary with the locality $k$. 
A numerical route can be taken to explore such behaviours for particular cases. Here we ask what kind of $H_{QS}$ yields the smallest spectral diameter $\Delta_{h}$, when $H_T$  is a Hamiltonian with $k$-body nearest neighbour interactions in the $z$-direction describing a one-dimensional array of qubits. For instance, for $k=3$ we force the target to be $H_T=C_3\sum_{j}\sigma_{j}^z\sigma_{j+1}^z\sigma_{j+2}^z$, 
where $C_3$ is just a normalization constant. As shown below, this problem can be cast as a semidefinite program: 
%An answer can be provided by solving a semidefinite program: 
\begin{proposition}\label{sdp_prop}
 The following semidefinite program minimizes $\Delta_h$ when $H_{QS}$ is $k'$-local and $H_T$ is a  $k$-local Hamiltonian with nearest neighbour interactions in the $z$-direction describing a one-dimensional array of $n$ qubits:
  \begin{equation} \label{eq:SPD_1}
  \begin{split}
    \min &\ \lambda_1-\lambda_0
    \\
    \text{\rm s.t.}&\ 
    \lambda_0 \preceq H_{QS} - H_T \preceq \lambda_1 \\
    &\ 
    \tr[H_{QS} \Lambda^{(j)}_i ] =0 \; \; \forall i \text{ and }
    \;\; \forall j>k'
    \\
    &\
    \tr[H_{QS} \Lambda^{(k')}_i ] \geq \beta \; \; \forall i 
    \\
    &\ 
    H_T
    =
    C_{k} \sum_{l=1}^n
    \sigma_l^z \ldots \sigma_{l+k-1}^z  
    % \\
    % &\ 
    % \tr[H_{T} \Lambda^{(j)}_i] =0 \; \; \forall i \text{ and }
    % \;\; \forall j>k
    % \\
    % &\ 
    % \tr[H_{T} \tilde{\Lambda}^{(k')}_i ] \geq C_{k'}
    % \; \; \forall i
    ,
  \end{split}
  \end{equation}
  where $C_{k}$ is a normalization constant, $\Lambda^{(j)}_i$ are the generators of all possible $j$-body interactions, and $\beta\neq 0$ is the strength of the $k'$-body interactions of $H_{QS}$.
  %$\tilde{\Lambda}^{(j)}_i$ are the generators of all possible $j$-body next neighbouring interactions of $z$ direction (see the Supplemental material)
  % \begin{equation}
  %     M_k^j 
  %     =
  %     \sigma_{\alpha}
  % \end{equation}
\end{proposition}
Note that the second constraint prevents $H_{QS}$  from having local terms with localities larger than  $k'$. Also, the third constraint forces $H_{QS}$ to present non-vanishing $k'$-body interactions.
Fig.~\ref{fig:diameter} shows the minimal spectral diameter $\Delta_{h}$  when $H_T$ is a Hamiltonian with $k$-body nearest neighbour  interactions ($k=3,4,5$) describing an array of 5 qubits, and $H_{QS}$ has only $2$-local interactions of strength larger than or equal to $\beta=0.01$. Periodic boundary conditions are enforced on $H_T$. The minimal value of $\Delta_{h}$ decreases when $H_T$ is less local, suggesting that such kind of target Hamiltonians are
% indicating that such kind of target Hamiltonians are \maria{more likely to be} 
simulated to a better precision the less local they are. %\maria{Further analytical research would need to be conducted in order to deeply understand why the spectral diameter of $h$ decreases with the locality of $H_T$ in this scenario.}

% In the settings where $b_2(h,t)< b_1(h,t)$, it holds that $1-\epsilon^{*}>0$ if and only if $t\norm{h}_{\text{HS}}<1$. 
% Taking again a numerical route,  we have considered the same scenario as before: $H_T$ is a Hamiltonian with $k$-body nearest neighbour  interactions ($k=3,4,5$) describing an array of 5 qubits, and $H_{QS}$ has only $2$-local interactions of strength larger than or equal to $\beta=0.01$. We here obtain that, for fixed $t$, $b_2(h,t)$ remains constant with the locality of $H_T$. % Further analysis would need to be carried out to better understand the behaviour of $b_2(h,t)$ with respect to the locality of $H_T$. }  

%and this indicates that the error $\epsilon$ of simulation in the worst case is smaller for less local $H_T$. 
% Simply, less local Hamiltonians are simulated better with norm fixed. 
% Although this sounds counterintuitive, 
%Although $H_T$ is restricted to neighbouring interactions, this result reveals the relation between locality and spectral diameter. 

% Figure~\ref{fig:diameter} also displays the spectral diameters of the $k$-body next neighbouring interactions of $z$ direction with norm fixed for $k=3,4,5$ (see the squares), and these spectral diameters also decrease for larger $k$.
% At $k=5$, the spectral diameter matches the minimised $\Delta_{H_{QS}-H_T}$, which means that the best $H_T$ to minimise $\Delta_{H_{QS}-H_T}$ corresponds to the 5-body next neighbouring interactions and the best $H_{QS}$ is 0. \\

\begin{figure}
\centering
\includegraphics[width=0.65\linewidth]{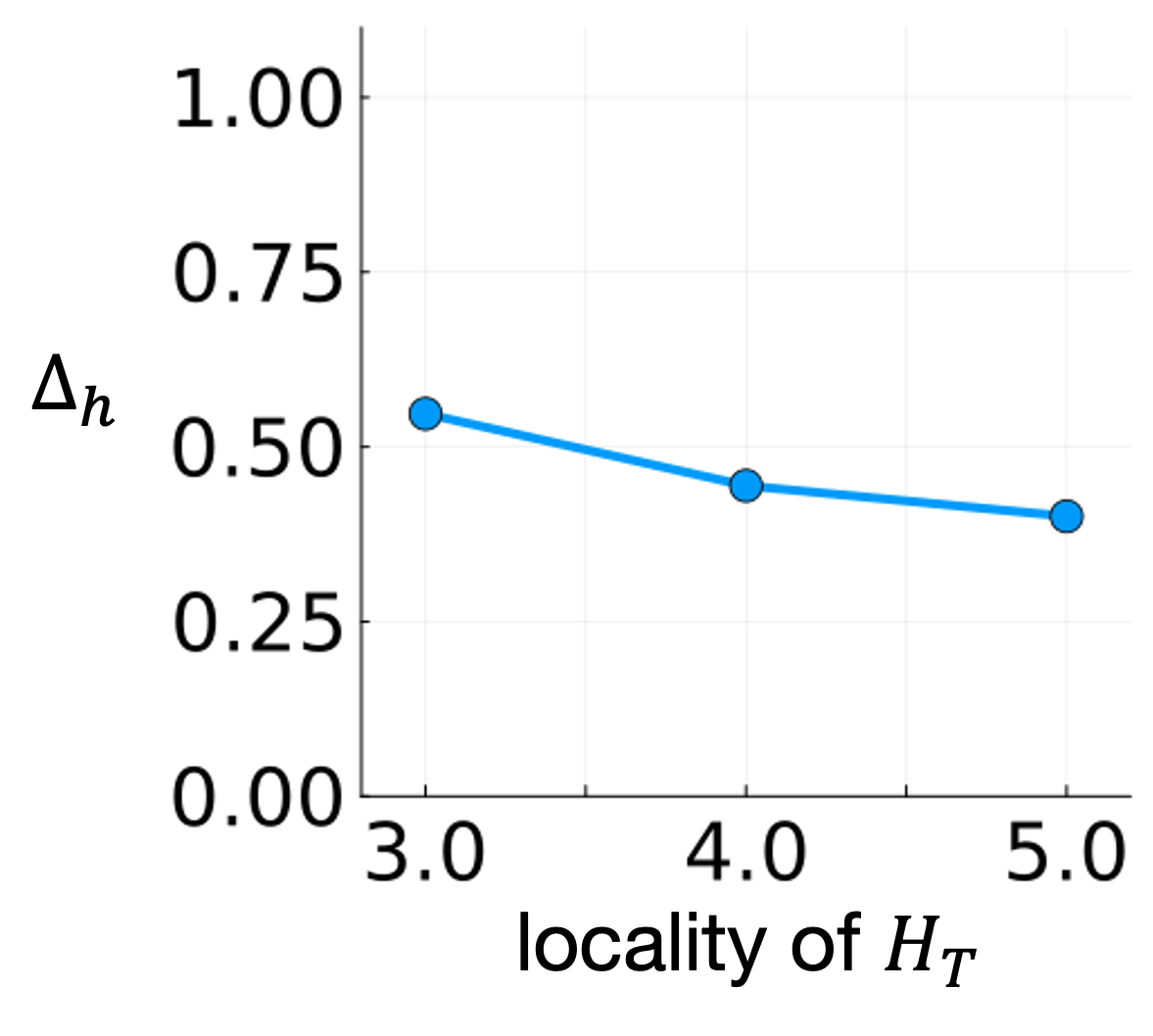}
\caption{
Spectral diameter $\Delta_{h}$ minimised by the SDP~\eqref{eq:SPD_1} as a function of the locality of $H_T$.
}
\label{fig:diameter}
\end{figure}

Fixing now $\Delta_h$ in \mbox{$t\Delta_h<2$}, 
% \mbox{$t\Delta_h<\log 3\approx 0.47$}, or $\norm{h}_\text{HS}$ in $t \norm{h}_\text{HS} <1$,
simulators are expectedly more precise when restricting to short times.
% we see that simulations are \maria{likely to be} more precise when restricting to short times, \maria{as one would expect.} 
%In this section 
Here, we provide a  program to find the $k'$-local Hamiltonian $H_{QS}$ that best simulates a target $k$-local Hamiltonian $H_T$ at a given state $\ket{\psi}$ for short times (in the precise sense that terms in $\mathcal{O}(t^2)$ are negligible):  
 \begin{proposition}\label{systematic}
 The following  program yields  the minimum HS norm between states \mbox{$\ket{\psi(t)}=e^{-itH_{T}}\ket{\psi_0}$} and $\ket{\phi(t)}=e^{-itH_{QS}}\ket{\psi_0}$ for state $\ket{\psi_0}$ and time $t\ll 1$, where $H_T$ is a given $k$-local Hamiltonian and $H_{QS}$---the optimization variable---is $k'$-local:
  \begin{equation}\label{eq:mins_HSnorm}
  \begin{split}
	  \min &\ \norm{(\mathds{1}-itH_T)\ket{\psi_0}-(\mathds{1}-itH_{QS})\ket{\psi_0}}_{\text{HS}}\\
	   \text{\rm s.t.}&\  \Tr(H_{QS}\; \Lambda^{(j)}_i)=0 \; \; \forall i,
  \end{split}
  \end{equation}
  where $\Lambda^{(j)}_i$ are the generators of all possible $j$-body interactions, with $j>k'$.
\end{proposition}
Notice that the states $(\mathds{1}-itH_X)\ket{\psi_0}$ are not normalized. It is possible to find a solution for this conic program using state-of-the-art solvers like \cite{mosek}. In Fig. \ref{fig:shortimes} we show the minimum HS norm between states \mbox{$\ket{\psi(t)}=e^{-itH_{T}}\ket{\psi_0}$} and \mbox{$\ket{\phi(t)}=e^{-itH_{QS}}\ket{\psi_0}$}, for several paradigmatic initial states $\ket{\psi_0}$. Here, $H_T$ is a random \mbox{3-qubit} 3-local Hamiltonian in 1D, and $H_{QS}$ is a 3-qubit 2-local Hamiltonian in 1D.  %The plotted distance scales linearly with time for almost of the initial states except for W state. The evolution of W state can be reproduced exactly in this case, and further studies are required. 
%Curiously, we observe that the evolution of the W state can be reproduced exactly in this case, 
%reflecting fact that the two-body reduced density matrix of the W state is entangled.
% As expected, for the rest of the initial states the plotted distance scales linearly with time. 
% The simulator found for the W state produces the same evolution exactly in this case. 
The evolution of the W state can be reproduced exactly in this case. For the rest of the initial states, the plotted distance scales linearly with time in this short time regime. 
% , and their simulators are not perfect.
% This indicates that there is no 2-local simulator working well for this target Hamiltonian and these initial states.

\begin{figure}
\centering
\includegraphics[width=0.9\linewidth]{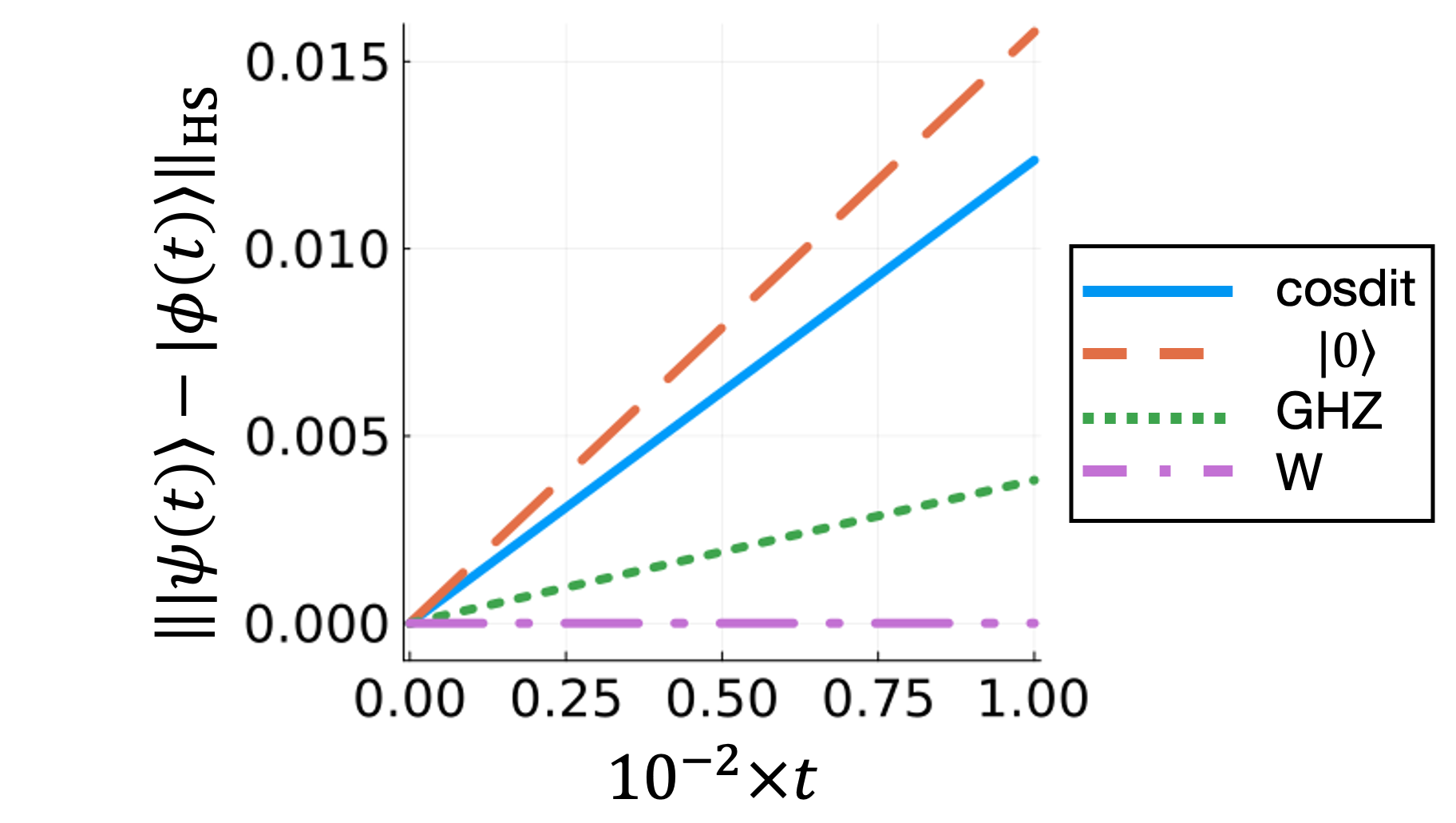}
\caption{Minimum HS norm between states \mbox{$\ket{\psi(t)}=e^{-itH_{T}}\ket{\psi_0}$} and $\ket{\phi(t)}=e^{-itH_{QS}}\ket{\psi_0}$, where $\ket{\psi_0}$ are different initial states (a cosdit $\frac{1}{\sqrt{8}}\sum_{i=1}^{8}\ket{i}$, the  $\ket{0}$ state, the GHZ state and the W state), $H_T$ is a random 3-qubit 3-local Hamiltonian in 1D, and $H_{QS}$ is a 3-qubit 2-local Hamiltonian in 1D.} 
\label{fig:shortimes}
\end{figure}

\section{Discussion}

We have explored the dynamic simulation of $k$-local Hamiltonians using more local Hamiltonians acting on the same  Hilbert space, an approach that aligns effectively with the present experimental limitations: as argued in the introduction, scaling up simulation platforms to accommodate remarkably larger simulator Hamiltonians is still a challenging task. 
In the exact simulation scenario, we have confirmed that some Hamiltonians with different localities can produce the same dynamics at certain subsets of states even if they do not commute. 
Further, we have analyzed the more realistic scenario where the dynamics is reproduced up to some error and provided a lower bound to the worst-case fidelity with which any Hamiltonian can simulate another one. Based on this, we have numerically shown evidence that the spectral diameter decreases as the locality increases, and this suggests that the simulation can be performed more precisely.
% Based on this, we have numerically shown that a one-dimensional qubit array with next-neighbour interactions \maria{is more likely to get} simulated to a higher precision the less local it is. 
Moreover, we have presented a program to find the Hamiltonian that best simulates  a given $k$-local Hamiltonian in the short-time regime and solved it for a particular physical scenario.

The relation between the worst-case fidelity of simulation and the spectral diameter of the corresponding connector has been unveiled. This has allowed us to study  dynamic simulatability without requiring diagonalization, which becomes computationally expensive when considering large systems. %demands computational resources unwanted for large systems. 
In turn, the relation between locality and spectral diameter needs to be further investigated in the line of \cite{Harrow2017}, which would help us understand how to better engineer quantum simulation settings. %and would help us understand how to eigneer eigenstates. 
Also,  further research could be conducted in the spirit of \cite{garciadiaz2022} to calculate  the relative volume of the simulatable sets of states, which would shed light on the potential of simulatability of each particular scenario.

\section*{Acknowledgements}
Discussions with \'Alvaro M. Alhambra and Tamara Kohler and comments from the referees are gratefully appreciated.
A.U. is financially supported by Maria de Maeztu project (Grant CEX2019-000918-M funded by MCIN/AEI/10.13039/501100011033). MGD was supported by grant PID2020-113523GB-I00, funded by
MCIN/AEI/10.13039/501100011033. We acknowledge support from the Spanish Agencia Estatal de Investigaci\'{o}n and the Ministerio de Ciencia e Innovaci\'{o}n (grants PID2019-107609GB-I00, PID2022-141283NB-I00). A.S. also acknowledges financial support from the European Commission QuantERA grant ExTRaQT (Spanish MICINN project PCI2022-132965), from the Spanish Ministry of Economic Affairs and Digital Transformation QUANTUM ENIA project call - Quantum Spain project, and by the European Union through the Recovery, Transformation and Resilience Plan - NextGenerationEU within the framework of the "Digital Spain 2026 Agenda", and by the Spanish MICINN with funding from European Union NextGenerationEU (PRTRC17.I1) and the Generalitat de Catalunya.

\bibliography{bibliografia}

\onecolumngrid

\appendix

\section{Illustration of exact simulation with a
Hamiltonian that is different from the target one}
%different Hamiltonian from target}
\label{app:illustration}

We aim at showing two examples of Hamiltonians with different localities which  generate the same dynamics.  
We consider Eq.~\eqref{eq:H3local} as the target Hamiltonian and two different simulator Hamiltonians as discussed below. 

\subsection{Example 1: a simulator that commutes with $H_T$} \label{app:example1}

As a simulator, we have considered the XXX model with periodic boundary condition, 
\begin{equation} \label{eq:XXXmodel}
    H_{QS}
    =
    J\sum_{j=1}^4
    \left(
    \sigma_x^j\sigma_x^{j+1}
    +
    \sigma_y^j\sigma_y^{j+1}
    + 
    \sigma_z^j\sigma_z^{j+1}
    \right)
    ,
\end{equation}
which commutes with $H_T$ for any $J$, $J_3$, and $h_x$, as mentioned in the main text. 
We take $J_3$ as the energy unit and here fix $h_x/J_3=1$, although this ratio can take
%but it can be 
any real value. 
We tune $J/J_3$ such that $h=H_{QS}-H_{T}$ has some degeneracy. 

In terms of how to find the shared eigenvectors, we refer to Appendix~\ref{app:shared_commuting}. 
We define the eigenvalues and the eigenvectors of $H_T$ as $\{\lambda_j^T\}_j$ and $\{\phi_j^{T}\}_j$, and the eigenvectors $\{\phi_j^{T}\}_j$ are tailored to diagonalise $H_{QS}$ as well. 
The eigenvalues $\{\lambda_j^{QS}\}_j$ of $H_{QS}$ are linearly dependent on $J$, and, by denoting the eigenvalues of $H_{QS}$ for $J/J_3=1$ with $\{\tilde{\lambda}_j^{QS}\}_{j}$, we have $\{\lambda_j^{QS}\}_j=J\{\tilde{\lambda}_j^{QS}\}_j$. By comparing $\{\lambda_j^{QS}\}_j$ and $\{\lambda_j^{T}\}_j$ and changing $J$, one can create some degeneracy in $h$.%=H_T-H_{QS}$. 

We here show %have found 
a few cases where $h$ has two degenerate eigenstates. %is degenerate in two eigenstates. %, and 
Fig.~\ref{fig:ex1} displays the parameter $J/J_3$ and the degenerate eigenstates  of $h$ of such cases. 
The same dynamics is generated by $H_T$ and by $H_{QS}$ as long as the initial state is spanned by %consists of
such eigenstates.

\begin{figure}[ht]
\centering
\includegraphics[width=.8\linewidth]{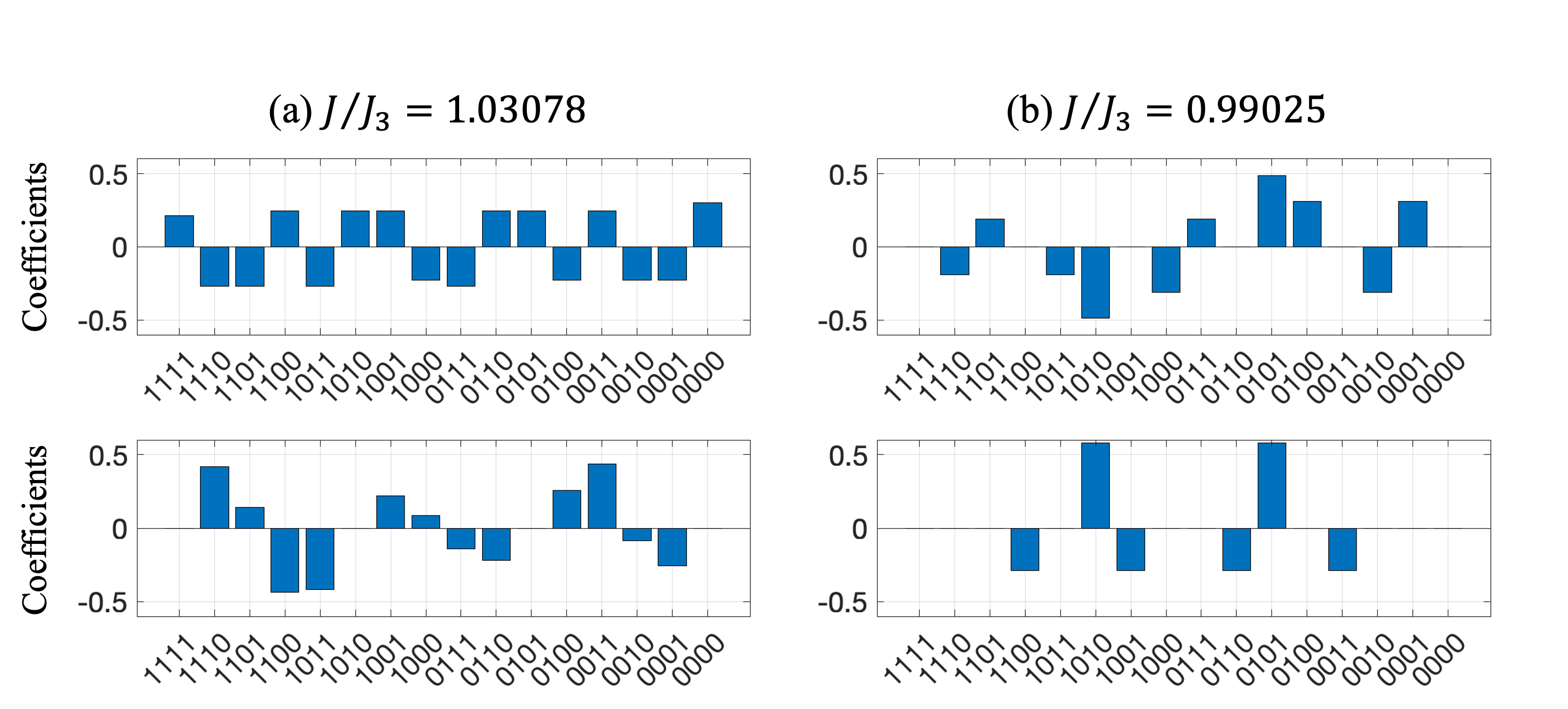}
\caption{
Coefficients of  two degenerate eigenstates of $h$ and  value of the parameter $J/J_3$. The x label denotes the spin basis, where ``0'' means spin down and ``1'' means spin up. 
}
\label{fig:ex1}
\end{figure}

\subsection{Example 2: a simulator that does not commute with $H_T$} \label{app:example2}

Consider a simulator $\tilde{H}_{QS}$ that does not commute with $H_T$, i.e. they may share only part of their eigenvectors. 
Our simulator is now a 2-local Heisenberg XXX model with 1-body terms and periodic boundary condition,
\begin{align}
     \tilde{H}_{QS}
     &=
     J \sum_{j=1}^4
     \left(
     \sigma_x^j\sigma_x^{j+1}
     +
     \sigma_y^j\sigma_y^{j+1}
     +
     \sigma_z^j\sigma_z^{j+1}
     \right)
     \nonumber\\
     &\quad
     +b_x\sum_{j=1}^4\sigma_x^j
     +b_y\sum_{j=1}^4\sigma_y^j
     +b_z\sum_{j=1}^4\sigma_z^j,  
\end{align}
where these 1-body terms prevent $\tilde{H}_{QS}$ from commuting with $H_T$. 
Again, we take $J_3$ as the energy unit. 

First, we search for the shared subspace. We refer the details to Appendix~\ref{app:shared_noncommuting}, but note that the size of the shared subspace corresponds to the nullity of the commutator, i.e., the number of vanishing
%zero  
eigenvalues of the commutator. We choose the parameter set $\{b_x,b_y,b_z\}$ such that %the size of 
the nullity is large, while $J$ is not involved in the eigenvalues of the commutator due to the fact that $[H_{QS},H_T]=0$. Thus, we take $\{b_x/J_3,b_y/J_3,b_z/J_3\}=\{-4,0,1\}$ and in this case the number of  shared eigenstates is 12. The set of the shared eigenstates is obtained by constructing a proper linear combination of the eigenstates of the commutator (see Appendices~\ref{app:shared_commuting}~\ref{app:shared_noncommuting}).

The only thing left is to create %make
degeneracy in $\tilde{h}=H_{QS}-\tilde{H}_{T}$ in the corresponding subspace as we do in Appendix~\ref{app:example1}. By tuning $J$, we have created some degeneracy in $\tilde{h}$. %, and
Fig.~\ref{fig:ex2} displays the parameter $J/J_3$ and the degenerate eigenstates of $\tilde{h}$ of such cases. 

\begin{figure}[ht]
\centering
\includegraphics[width=.4\linewidth]{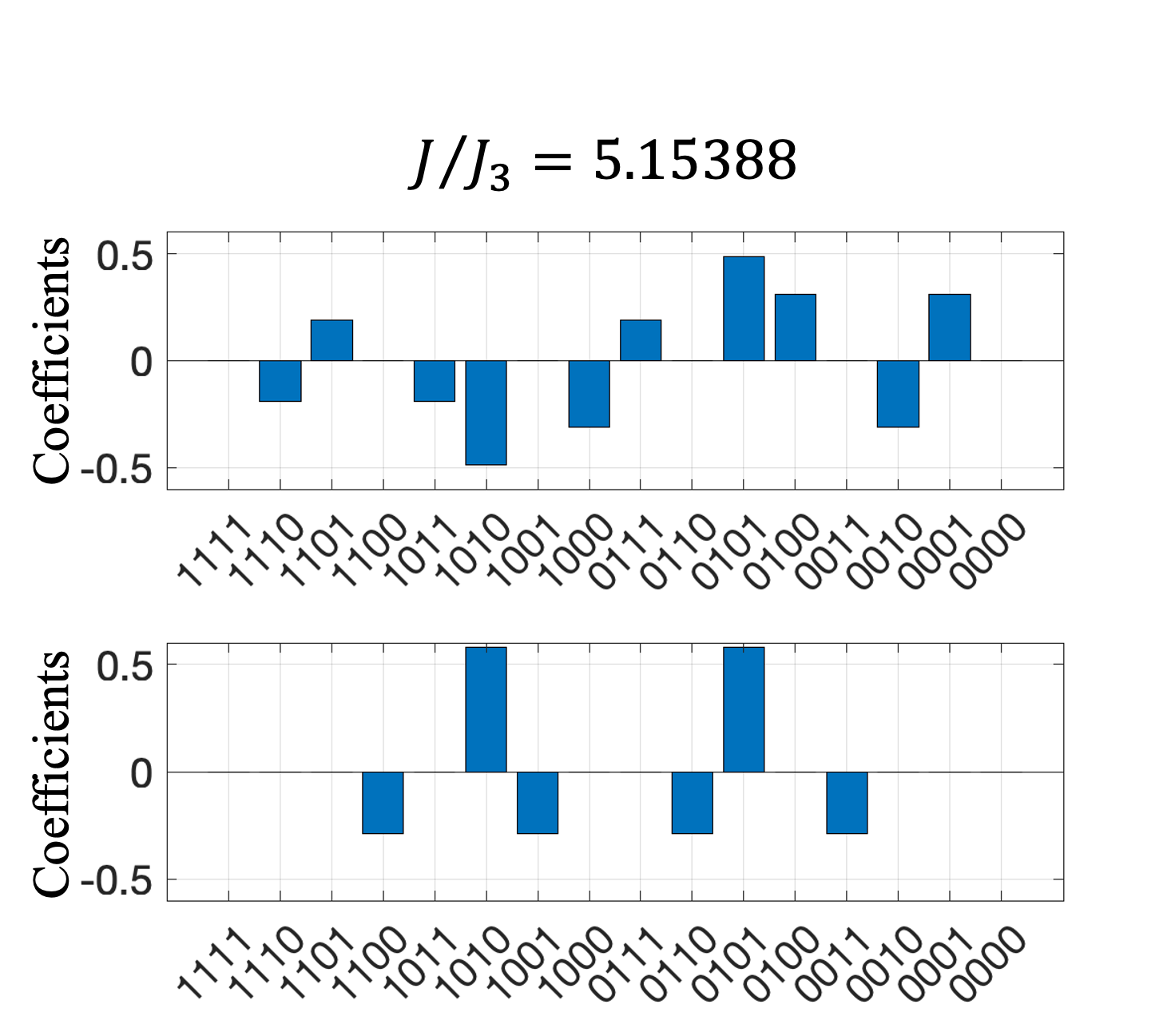}
\caption{
Coefficients of two eigenstates degenerate in $\tilde{h}$ and the value of parameter $J/J_3$. The x label denotes the spin basis, where ``0'' means down spin and ``1'' means up spin. 
}
\label{fig:ex2}
\end{figure}

\section{Shared eigenstates of two commuting Hamiltonians} \label{app:shared_commuting}

Two commuting Hamiltonians share all of their eigenstates. One can construct the shared eigenstates by using diagonalization, but there is an extra step to take if these Hamiltonians are both degenerate. Below, we explain how the shared eigenvectors of two commuting Hamiltonians are obtained. Our discussion is general but tailored for use in Appendix~\ref{app:illustration}. 

Suppose two commuting Hamiltonians $H_{\mathrm{A}}$ and $H_{\mathrm{B}}$.
We consider two cases: (i) the case where the Hamiltonian $H_{\mathrm{A}}$ is not degenerate and (ii) the case where the Hamiltonian $H_{\mathrm{A}}$ has one set of  degenerate eigenstates. 

Let us start with  case~(i). We define $\ket{\psi_n}$ as the $n$-the eigenstate of $H_{\mathrm{A}}$,
\begin{align}
    H_{\mathrm{A}}\ket{\psi_n}
    &=
    a_n \ket{\psi_n}
    .
\end{align}
We apply $H_{\mathrm{A}}H_{\mathrm{B}}$ on the state $\ket{\psi_n}$ and use $[H_{\mathrm{A}},H_{\mathrm{B}}]=0$, 
\begin{align}\label{eq:HaHbket}
    H_{\mathrm{A}}
    H_{\mathrm{B}}
    \ket{\psi_n}
    &=
    H_{\mathrm{B}}
    H_{\mathrm{A}}
    \ket{\psi_n}
    \nonumber\\
    &=
    a_n
    H_{\mathrm{B}}
    \ket{\psi_n}
    .
\end{align}
This means that, if $H_{\mathrm{B}}\ket{\psi_n}\neq 0$, $H_{\mathrm{B}}\ket{\psi_n}$ is an eigenstate of $H_{\mathrm{A}}$. 
Since $H_{\mathrm{A}}$ is not degenerate, $\ket{\psi_n}$ is only one eigenstate having $a_n$ as its eigenvalue. Thus, this eigenstate $H_{\mathrm{B}}\ket{\psi_n}$ is essentially the same state as $\ket{\psi_n}$, and they differ  only by a constant. We call this constant $b_n$, and therefore,
\begin{align}\label{eq:psineigenHB}
    H_{\mathrm{B}}
    \ket{\psi_n}
    &=
    b_n
    \ket{\psi_n}
    .
\end{align}
If $H_{\mathrm{B}}\ket{\psi_n}=0$, we say that $\ket{\psi_n}$ is an eigenstate of $H_{\mathrm{B}}$ with its eigenvalue $b_n=0$. 

Next, we consider  case~(ii) where $H_{\mathrm{A}}$ has a set of degenerate eigenstates. We cannot follow the same process as in case~(i) because, when we obtain Eq.~\eqref{eq:psineigenHB}, we rely on the fact that all the eigenvalues are different. 
Suppose that the $n$-th eigenvalue $a_n$ is $g$-fold degenerate. Let us define $\ket{\psi_n^r}$ for $r=1,2,\ldots,g$ as its $g$-fold degenerate eigenstates: $H_{\mathrm{A}}\ket{\psi_n^r}=a_n\ket{\psi_n^r}$.
The set of $\{\ket{\psi_n^r}\}_r$ is chosen such that they are orthonormalised via the Gram–Schmidt process.
The state $H_{\mathrm{B}}\ket{\psi_n^r}$ can be expanded in a set of the degenerate eigenstates $\{\ket{\psi_n^r}\}_r$ as
\begin{align}
    H_{\mathrm{B}}
    \ket{\psi_n^r}
    &=
    \sum_{s=1}^{g}
    c_{r,s}\ket{\psi_n^s}
\end{align}
with a set of coefficients $c_{r,s}$, which is given by $\bra{\psi_n^{s}}H_{\mathrm{B}}\ket{\psi_n^r}$ because
\begin{align}
    \bra{\psi_n^{r'}}
    H_{\mathrm{B}}
    \ket{\psi_n^r}
    &=
    \sum_{s=1}^{g}
    c_{r,s}\braket{\psi_n^{r'}}{\psi_n^s}
    \nonumber\\
    &=
    c_{r,r'}
    .
\end{align}
Consider superposition between the degenerate states, $\sum_{r=1}^g d_r \ket{\psi_n^r}$ with $d_r$ coefficient. We apply $H_{\mathrm{B}}$ to this state, 
\begin{align}
    H_{\mathrm{B}}
    \sum_{r=1}^g
    d_r
    \ket{\psi_n^r}
    &=
    \sum_{r=1}^g
    \sum_{s=1}^{g}
    d_r
    c_{r,s}\ket{\psi_n^s}
    \nonumber\\
    &=
    \sum_{s=1}^{g}
    \left(
    \sum_{r=1}^g
    c_{r,s}
    d_r
    \right)
    \ket{\psi_n^s}
\end{align}
The state $\sum_{r=1}^g d_r\ket{\psi_n^r}$ is an eigenstate of $H_{\mathrm{B}}$ with its eigenvalue $b_n$ if 
\begin{align}
    \sum_{r=1}^g
    c_{r,s}
    d_r
    &=
    b_n
    d_s
    .
\end{align}
Considering all the elements, this can be represented in matrix form, 
\begin{align}
    c \vec{d}
    &=
    b_n \vec{d}
    .
\end{align}
The set of $b_n$ and $\{d_r\}_r$ that satisfy the above equation is given by the eigenvalues and the eigenvectors of matrix $c$, respectively. We define the eigenvalues as $b_n^{(k)}$ and the eigenvectors as $\{d_r^{(k)}\}_r$ for $k=1,2,\ldots,g$. 
Therefore, the eigenvector of $H_{\mathrm{B}}$ with %its 
eigenvalue $b_n^{(k)}$ is given by $\sum_{r=1}^g d_r^{(k)} \ket{\psi_n^r}$. Of course, the eigenvalue of $H_{\mathrm{A}}$ with %its 
eigenvector $\sum_{r=1}^g d_r^{(k)} \ket{\psi_n^r}$ is $a_n$.

% \maria{The point here is that $H_{\mathrm{A}}$ does not distinguish $\ket{\psi_n^r}$ for different $r$ since their eigenvalues are the same, but $H_{\mathrm{B}}$ may (I don't understand this sentence).} 
The point here is that any linear combination of $\ket{\psi_n^r}$ for all $r$ can diagonlise $H_{\mathrm{A}}$ but may not diagonlise $H_{\mathrm{B}}$.
Therefore, it is necessary to find a linear combination of $\ket{\psi_n^r}$ that is an eigenstate of $H_{\mathrm{B}}$.
If $H_A$ has multiple sets of degenerate states, we take this procedure for each set. 

Note that the degeneracy of $H_{\mathrm{B}}$ does not matter in this process. Thus, if it is known that $H_{\mathrm{A}}$ is degenerate but $H_{\mathrm{B}}$ is not, it is efficient to use the eigenvectors of $H_{\mathrm{B}}$ as the shared eigenvectors instead. 

\section{Shared subspace of two non-commuting Hamiltonians} \label{app:shared_noncommuting}

Even two non-commuting Hamiltonians may share some of their eigenstates. 
We explain how to find such shared eigenvectors below, and our discussion is general but tailored
for use in Appendix~\ref{app:example2}.
First, consider two non-commuting Hamiltonians $\hat{H}_{\mathrm{A}}$, $\hat{H}_{\mathrm{B}}$  and calculate the commutator,
\begin{align}
    \left[
    \hat{H}_{\mathrm{A}},\hat{H}_{\mathrm{B}}
    \right]
    &=
    C
    ,
\end{align}
where $C$ is a traceless matrix. Denoting a set of eigenvectors of the commutator with $\{\ket{\phi_j}\}_j$, we can have
\begin{align}
    \sum_{j=1}^N
    \ketbra{\phi_j}{\phi_j}
    \left[
    \hat{H}_{\mathrm{A}},\hat{H}_{\mathrm{B}}
    \right]
    \ketbra{\phi_j}{\phi_j}
    &=
    0^{\oplus r}\oplus \tilde{C}
\end{align}
with $N$ the system size, $r\geq 0$ some integer and $\tilde{C}$ a diagonal and traceless matrix. The subspace spanned by the first $r$ eigenvectors $\{\ket{\phi_j}\}_{j=1,\ldots,r}$ is the shared subspace, where the commutator is zero. 

The first $r$ eigenvectors $\{\ket{\phi_j}\}_{j=1,\ldots,r}$ may not diagonalise $\hat{H}_{\mathrm{A}}$ and $\hat{H}_{\mathrm{B}}$ in the subspace. It is the same reason why  the eigenvectors of commuting degenerate Hamiltonians may not diagonalise one of the Hamiltonians as mentioned in Appendix~\ref{app:shared_commuting}, and essentially these eigenvectors $\{\ket{\phi_j}\}_{j=1,\ldots,r}$ are degenerate in the commutator. One can find a linear combination of $\{\ket{\phi_j}\}_{j=1,\ldots,r}$ that diagonalises $\hat{H}_{\mathrm{A}}$ and $\hat{H}_{\mathrm{B}}$ in the subspace by following the procedure introduced in Appendix~\ref{app:shared_commuting}.  

\section{Necessary condition for a state to be an
eigenstate of both a $k$-local and a $k'$-local
Hamiltonian}\label{necessary_eig}

Consider the set of $k$-local Hamiltonians, with basis $\{\Lambda_i\}$ and $\Lambda_i\in {\text{Herm}}(\cal{H})$. For a given state $\ket{\psi}$,  the corresponding correlation matrix is defined as 
\begin{equation}
    M_{ij}^{(k,\psi)}=\frac{1}{2}\bra{\psi}\{\Lambda_i,\Lambda_j\}\ket{\psi}-\bra{\psi}\Lambda_i\ket{\psi}\bra{\psi}\Lambda_j\ket{\psi}. 
\end{equation}
If the kernel of $M^{(k,\psi)}$ is not empty, then there exists at least one $k$-local Hamiltonian that is parent to $\ket{\psi}$~\cite{Qi2019}. Since $M^{(k,\psi)}$ is Hermitian, it will have a non-empty kernel  if and only if $\det M^{(k,\psi)}=0$. 

Now, recall that the set $\{\Lambda_i\}$ contains the generators of every possible $j$-body interaction, with $j\leq k$. Therefore, it contains the generators of all  $k$-body interactions and the generators of all $k'$-body interactions, with $k'<k$. This implies that $M^{(k,\psi)}$ is a block matrix of the form
\begin{equation}
  M^{(k,\psi)}=  
  \begin{pmatrix}
  \begin{matrix}
  {\cal{M}}_{kk}
  \end{matrix}
  & \rvline & {\cal{M}}_{kk'} \\
\hline
  {\cal{M}}_{k'k} & \rvline &
  \begin{matrix}
  {\cal{M}}_{k'k'}
  \end{matrix}
\end{pmatrix},
\end{equation}
where ${\cal M}_{kk'}$ encompasses the correlations between the generators of the \mbox{$k$-body} and the $k'$-body interactions. The following lemma holds:

\begin{lemma}
If a state $\ket{\psi}$ is an eigenstate of both a $k$-local (with non-vanishing $k'$-local terms) and a $k'$-local Hamiltonian ($k'<k$), then
\begin{equation}
    \sum_{r=1}^{N-1}\sum_{\alpha,\beta}(-1)^{s(\alpha)+s(\beta)}\det A[\alpha|\beta]\det B(\alpha|\beta)=-\det B,
\end{equation}
where  $A={\cal M}_{k'k'}$, $B=-{\cal M}_{k'k}{\cal M}_{kk}^{-1}{\cal M}_{kk'}$, and \mbox{$N=\dim A=\dim B$}. For a particular $r$, the inner sum is taken over all strictly increasing integer sequences $\alpha$ and $\beta$ of length $r$, from 1 to $N$. $A[\alpha|\beta]$ is the $r$-square submatrix of $A$ lying in rows $\alpha$ and columns $\beta$,  $B(\alpha|\beta)$ is the $(N-r)$-square submatrix of $B$ lying in rows complementary to $\alpha$ and columns complementary to $\beta$, and $s(\alpha)$ is the sum of the integers in $\alpha$.
\end{lemma}
\begin{proof}
Since we do not want $\ket{\psi}$ to be an eigenstate only of a $k$-local Hamiltonian with vanishing $k'$-local terms, it is sound to assume that $\det {\cal M}_{kk}\neq 0$. Under this condition, 
\begin{equation}
    \det M^{(k,\psi)}=\det {\cal M}_{kk}\det({\cal M}_{k'k'}-{\cal M}_{k'k}{\cal M}_{kk}^{-1}{\cal M}_{kk'}),
\end{equation}
which is zero when 
\begin{equation}\label{det0}
  \det({\cal M}_{k'k'}-{\cal M}_{k'k}{\cal M}_{kk}^{-1}{\cal M}_{kk'})=0.  
\end{equation}
Now recall that, for $N$-square matrices $A$ and $B$ \cite{marcus1975finite}, 
\begin{eqnarray*}
    \det(A+B)&=&\det A+\det B +\sum_{r=1}^{N-1}(-1)^{s(\alpha)+s(\beta)}\det A[\alpha|\beta]\det B(\alpha|\beta).    
\end{eqnarray*}
Define $A={\cal M}_{k'k'}$ and $B=-{\cal M}_{k'k}{\cal M}_{kk}^{-1}{\cal M}_{kk'}$.  If we want $\ket{\psi}$ to also be an eigenstate of a $k'$-local Hamiltonian,  we must impose \mbox{$\det A=0$}. Together with Eq.(\ref{det0}), this yields 
\begin{equation*}
    \sum_{r=1}^{N-1}\sum_{\alpha,\beta}(-1)^{s(\alpha)+s(\beta)}\det A[\alpha|\beta]\det B(\alpha|\beta)=-\det B,
\end{equation*}
and the proof is completed. 
\end{proof}

%%%%%%%%%%%%%%%%%%%%%
% \section{Comparison between the bounds of Theorem \ref{thm_nonexact}}\label{comparison_bounds}

\section{A weaker version of Theorem \ref{thm_nonexact}}\label{comparison_bounds}

The following theorem can be derived:
\begin{theorem}\label{thm_nonexact}
\cite{lidar2008,referee} Every Hamiltonian $H_{QS}$ \mbox{$\epsilon^*$-simulates} any target Hamiltonian $H_T$ at any state $\ket{\psi}$ and time $t$, with 
\begin{equation}
    \epsilon^*=\min\left[ 1,\dfrac{1}{2}\left(e^{t\Delta_{h}}-1\right), t\norm{h}_\text{HS}\right].
\end{equation}
Here, $\Delta_{h}$ is the spectral diameter of the connector, \mbox{$h:=H_{QS}-H_{T}$}, i.e., \mbox{$\Delta_{h}=\max_{ij}|\lambda_i-\lambda_j|$}, where $\{\lambda_i\}_{i=1}^N$ are the eigenvalues of $h$. 
\end{theorem}

\begin{proof}
The bound $\epsilon^*=\min\left[ 1,\dfrac{1}{2}\left(e^{t\Delta_{h}}-1\right)\right]$ is based on a beautiful result derived in \cite{lidar2008}. There, the authors show that the trace distance between two states evolved under different time-dependent Hamiltonians, $\Omega(t)$ and $\Omega^0(t)$, fulfills \mbox{$D[\rho(t),\rho^0(t)]\leq\min \left[1,\frac{1}{2}(e^{t\Delta \lambda(t)}-1) \right]$}, where \mbox{$\rho(t)=e^{-it[\Omega(t),\;\cdot\;]}\rho(0)$}, $\rho^0(t)=e^{-it[\Omega^0(t),\;\cdot \;]}\rho(0)$, and $\Delta \lambda(t)$ is the spectral diameter of \mbox{$\Omega(t)-\Omega^0(t)$}. The derivation  is completed by applying that the fidelity between two pure states, $F(\ket{\psi},\ket{\phi})=|\braket{\psi}{\phi}|$, is lower bounded as \mbox{$F(\ket{\psi},\ket{\phi})\geq 1-D(\ketbra{\psi}{\psi},\ketbra{\phi}{\phi})$}, with \mbox{$D(\sigma,\tilde{\sigma})=\frac{1}{2}\tr |\sigma-\tilde{\sigma}|$}.

We now proceed to derive $\epsilon^*=\min\left[ 1,t\norm{h}_\text{HS}\right]$ \cite{referee}. 
First, we have that 
\begin{eqnarray}
        &&D(e^{-itH_T}\ketbra{\psi}{\psi}e^{itH_T},e^{-itH_{QS}}\ketbra{\psi}{\psi}e^{itH_{QS}})  \nonumber \\
        &\stackrel{(a)}{\leq}& \norm{(e^{-itH_T}-e^{-itH_{QS}})\ket{\psi}}_2 \nonumber\\
         &\stackrel{(b)}{\leq}& \norm{e^{-itH_T}-e^{-itH_{QS}}}_{\text{HS}} \nonumber\\
          &\stackrel{(c)}{\leq}& t\norm{h}_{\text{HS}}.\nonumber
\end{eqnarray}
        In (a) we have used that, for pure states $\ket{\psi}$ and $\ket{\phi}$, it holds that $D(\ketbra{\psi}{\psi},\ketbra{\phi}{\phi})\leq \norm{\ket{\psi}-\ket{\phi}}_2$, where $\norm{\cdot}_2$ is the Euclidean norm of a vector.  In (b) we have employed the Cauchy-Schwarz inequality and the fact that $\ket{\psi}$ is normalized. In (c) we have first applied the triangle inequality to a telescoping sum. For $A:=-itH_T$ and $B:=-itH_{QS}$ we get 
        \begin{eqnarray}
           && \norm{e^A-e^B}_\text{HS} \nonumber\\
           &=& \norm{\sum_{k=1}^m e^{(k-1)A/m}(e^{A/m}-e^{B/m})e^{(m-k)B/m}}_\text{HS}\nonumber\\
            &\leq& \sum_{k=1}^m \norm{e^{(k-1)A/m}(e^{A/m}-e^{B/m})e^{(m-k)B/m}}_\text{HS} \nonumber\\
            &=& m \norm{e^{A/m}-e^{B/m}}_\text{HS},
        \end{eqnarray}
for any  $m\in \mathbb{N}$. Note that the last equality holds since the HS norm is unitarily invariant. Finally, expanding $e^{A/m}$ and $e^{B/m}$ into power series and letting $m\rightarrow \infty$ immediately leads to $\norm{e^A-e^B}_\text{HS}\leq \norm{A-B}_\text{HS}$. Apply  \mbox{$F(\ket{\psi},\ket{\phi})\geq 1-D(\ketbra{\psi}{\psi},\ketbra{\phi}{\phi})$} and the proof is finished.
%$|\bra{\psi}e^{itH_{QS}}e^{-itH_{T}}\ket{\psi}|\geq 1-t\norm{h}_{\text{HS}}$.
\end{proof}

Notice that this theorem is weaker than Theorem \ref{thm_nonexact}, since $\frac{1}{2}(e^{t\Delta_h}-1)\geq \frac{t\Delta_h}{2}$, and $\norm{h}_{\text{HS}}\geq \norm{h}_2$. 

What is more, neither of the bounds, $b_1(h,t):=\frac{1}{2}\left(e^{t\Delta_{h}}-1\right)$ nor $b_2(h,t):=t\norm{h}_\text{HS}$, can universally upper-bound the other.
We hereby present two examples of connectors leading to different behaviours of $b_1(h,t)$ with respect to $b_2(h,t)$.

First, we consider the randomly generated 3-local target Hamiltonian and the 2-local QS Hamiltonian of Fig. \ref{fig:nonuniversal}. As shown in the figure, as long as $b_1(h,t)$  and $b_2(h,t)$ are smaller than 1, $b_1(h,t)$ is smaller than $b_2(h,t)$ for times $t<0.6$, and greater than it otherwise.

Secondly, we take the 3-local target Hamiltonian of Proposition \ref{sdp_prop} and a 2-local QS Hamiltonian (see Fig. \ref{fig:universal}). As shown in the figure, as long as $b_1(h,t)$  and $b_2(h,t)$ are smaller than 1, $b_1(h,t)$ is always smaller than $b_2(h,t)$.

\begin{figure}[th]
  \centering
  \includegraphics[width=0.9\textwidth]{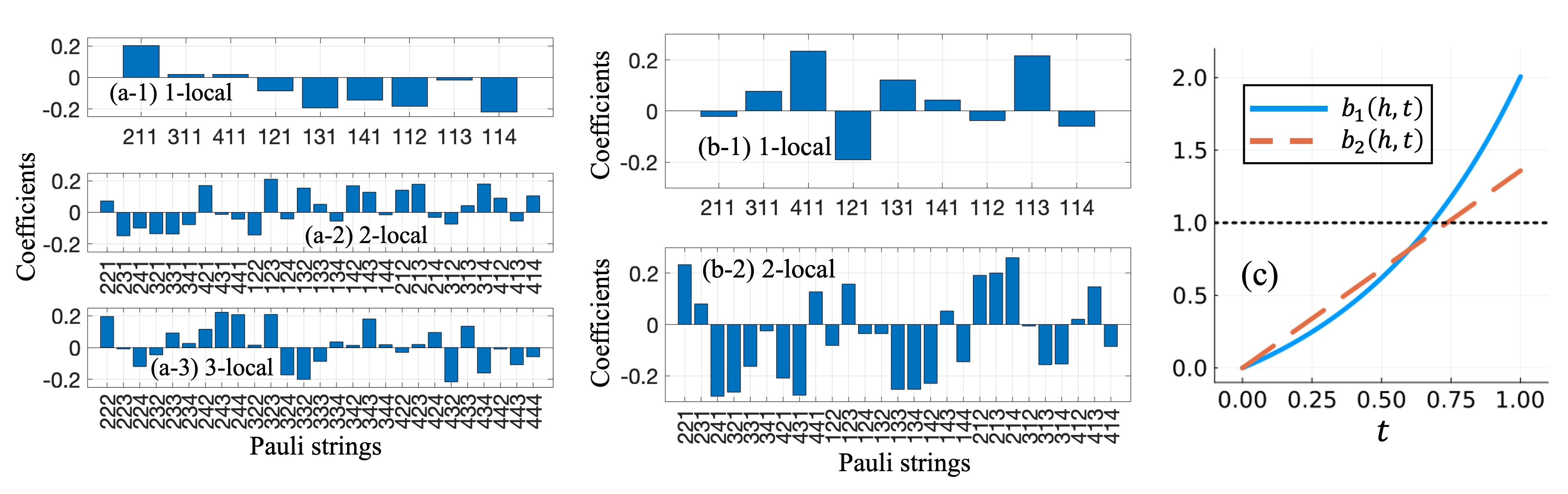}
  % \subfigure[$\alpha_1$]{\includegraphics[width=0.7\textwidth]{figures/plot_bounds_Ht_3local.eps}}
  
  % \subfigure[$\beta_1$]{\includegraphics[width=0.7\textwidth]{figures/plot_bounds_Hqs_2local.eps}}
  
  % \subfigure[$\gamma_1$]{\includegraphics[width=0.7\textwidth]{figures/bounds_nonuniversal.eps}}
  \caption{(a)~3-qubit $H_T$, (b)~3-qubit $H_{QS}$ and (c)~behaviour of the corresponding quantities $b_1(h,t)$ and $b_2(h,t)$. If the Hamiltonian is expressed as $H=\sum_l c_l P_l$, with $P_l\in\{\mathds{1},\sigma_x,\sigma_y,\sigma_z\}^{\otimes 3}$ a Pauli string, the x-axis represents the Pauli strings $ijk\equiv P_l$, with $1\equiv \mathds{1}$, $2\equiv \sigma_x$, $3\equiv \sigma_y$, and $4\equiv \sigma_z$. The corresponding coefficients, $c_l$, are shown in the y-axis.}
  \label{fig:nonuniversal}
\end{figure}

\begin{figure}[th]
  \centering
  \includegraphics[width=0.9\textwidth]{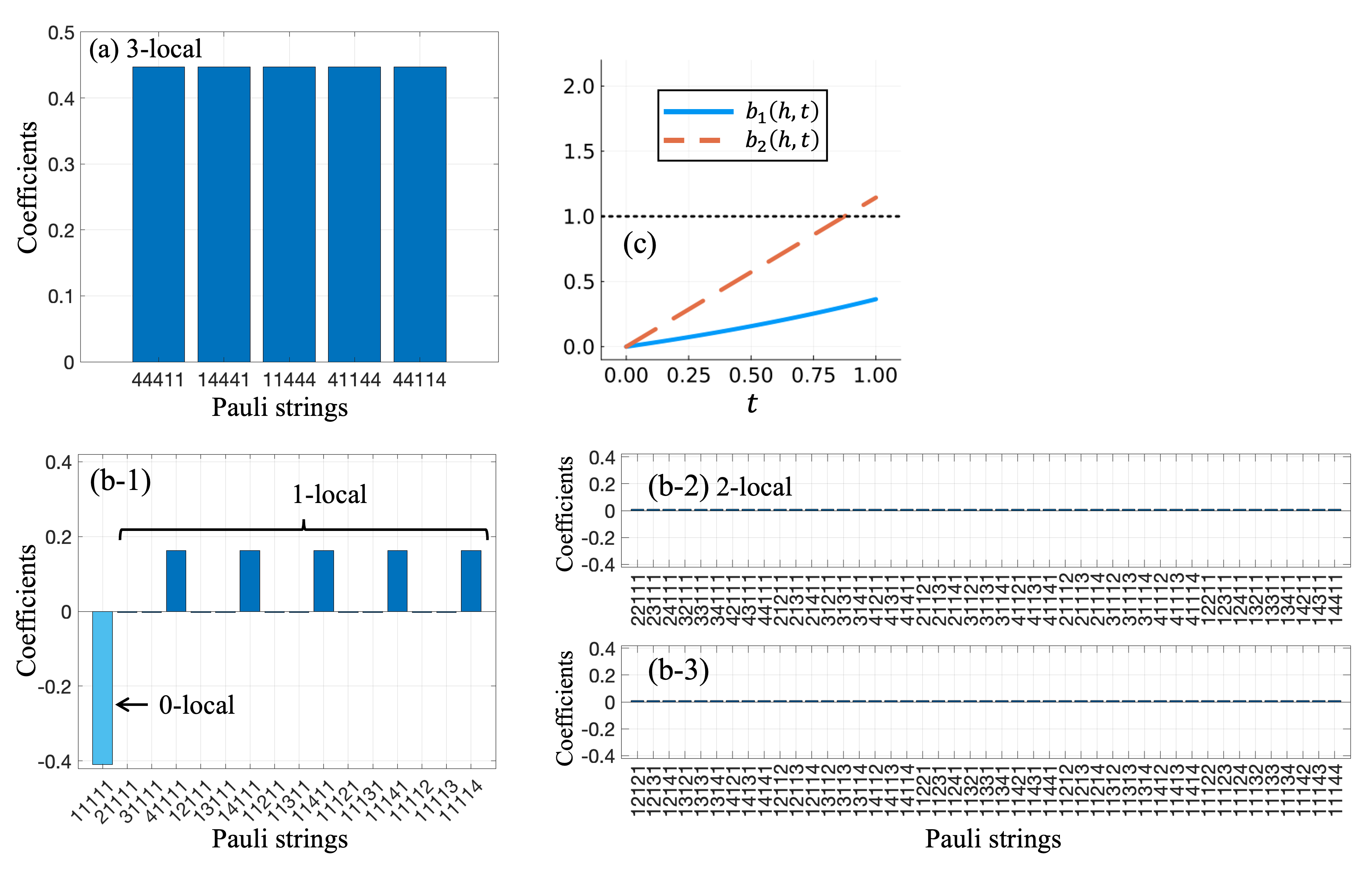}
  % \subfigure[$\alpha_2$]{\includegraphics[width=0.7\textwidth]{figures/plot_bounds_Ht_5qubits.eps}}
  % \subfigure[$\beta_2$]{\includegraphics[width=1\textwidth]{figures/plot_bounds_Hqs_5qubits.eps}}
  % \subfigure[$\gamma_2$]{\includegraphics[width=0.7\textwidth]{figures/bounds_universal.eps}}
\caption{(a)~5-qubit $H_T$, (b)~5-qubit $H_{QS}$ and (c)~behaviour of the corresponding quantities $b_1(h,t)$ and $b_2(h,t)$. If the Hamiltonian is expressed as $H=\sum_l c_l P_l$, with $P_l\in\{\mathds{1},\sigma_x,\sigma_y,\sigma_z\}^{\otimes 5}$ a Pauli string, the x-axis represents the Pauli strings $ijk\equiv P_l$, with $1\equiv \mathds{1}$, $2\equiv \sigma_x$, $3\equiv \sigma_y$, and $4\equiv \sigma_z$. The corresponding coefficients, $c_l$, are shown in the y-axis.
% In (b), all the 2-local terms in $H_{QS}$ are 0.01 and so small that they are not represented. 
}
  \label{fig:universal}
\end{figure}

\end{document}